\documentclass[final,1p]{elsarticle}

\usepackage{epsfig,color}
\usepackage{amsmath}
\usepackage{amstext}
\usepackage{amsfonts}
\usepackage{amssymb,amscd}
\usepackage{bm,bbm}

\newtheorem{thm}{Theorem}
\newtheorem{defi}{Definition}
\newtheorem{lemma}{Lemma}
\newtheorem{corollary}{Corollary}
\newtheorem{propo}{Proposition}\newtheorem{proposition}{Proposition}

\newenvironment{proof}[1][Proof]{\noindent\textbf{#1.} }{\ \rule{0.5em}{0.5em}}

\newcommand{\conv}{\rm conv}
\newcommand{\bo}{\partial}

\newcommand{\tr}{\rm{tr}}\newcommand{\Tr}{\rm{Tr}}

\newcommand{\1}{\mathbf 1}

\newcommand{\tv}{\tilde{v}}
\newcommand{\tu}{\tilde{u}}

\newcommand{\A}{\mathcal A}
\newcommand{\B}{\mathcal B}
\newcommand{\C}{\mathbb C}
\newcommand{\HH}{\mathcal H}

\newcommand{\LL}{\mathcal L}\newcommand{\LLL}{\LL\LL}
\newcommand{\N}{\mathbb N}
\newcommand{\PP}{\mathcal P}
\newcommand{\Q}{\mathcal Q}
\newcommand{\R}{\mathbb R}
\newcommand{\SSS}{\mathcal M}

\newcommand{\Ga}{\Gamma}
\newcommand{\vp}{\varphi}
\newcommand{\si}{\sigma}
\newcommand{\Om}{\Omega}

\begin{document}
\noindent
\bibliographystyle{plain}

\begin{frontmatter}

\title
{Joint numerical ranges, quantum maps,\\
 and joint numerical shadows}
\author[umk,impan]{Eugene Gutkin}
\author[uj,cft]{Karol~{\.Z}yczkowski}
\address[umk]{Department of Mathematics, Nicolaus Copernicus University,
Chopina 12/18, 87-100 Toru{\'n}, Poland}
\address[impan]{Institute of Mathematics, Polish Academy of Sciences,
\'Sniadeckich 8, 00-956 War\-sza\-wa 10, Poland}
\address[uj]{Instytut Fizyki im. Smoluchowskiego, Uniwersytet
Jagiello{\'n}ski, Reymonta 4, 30-059 Krak{\'o}w, Poland }
\address[cft]{Centrum Fizyki Teoretycznej, Polska Akademia Nauk, Aleja Lotnik{\'o}w
32/44, 02-668 War\-sza\-wa, Poland}
%

\begin{abstract}
We associate with a $k$-tuple of hermitian $N\times N$ matrices a probability measure
on $\R^k$ supported on their joint numerical range: The joint numerical shadow of these matrices.
When $k=2$ we recover the numerical range and the numerical shadow of the complex matrix corresponding to
a pair of hermitian matrices. We apply this material to the theory of quantum information.
Thus, we show that quantum maps on the set of quantum states defined
by Kraus operators satisfying  the identity resolution assumption shrink joint numerical ranges.


\end{abstract}





\date{\today}


\end{frontmatter}


\section{Introduction}
\label{intro}

Let ${\mathcal H}_N$ be the $N$-dimensional Hilbert space with the scalar product $\langle \vp|\psi\rangle$,
and let $A$ be an operator on ${\mathcal H}_N$. Its {\em numerical range} (also called the {\em field of values})
$W(A)\subset\C$ is the set of numbers $z=\langle \psi|A|\psi\rangle$, where $\langle\psi|\psi\rangle=1$ \cite{HJ2,GR97}.
A crucial fact, conjectured by O. Toeplitz in 1918
and proved by F. Hausdorff in 1919 \cite{Toeplitz,Hausdorff} is that $W(A)$ is convex.
See \cite{Gu04} for an exposition of \cite{Toeplitz,Hausdorff} from a modern perspective. 
If $A$ is a normal operator,\footnote{Recall that an operator $A$ is normal if $AA^*=A^*A$.} 
then $W(A)$ is the convex hull of the spectrum of $A$, hence a convex polygon. For generic $A$ the boundary
$\bo W(A)$ is smooth \cite{JAG98}.
If  $N=2$, $\bo W(A)$ is a (possibly degenerate) ellipse \cite{Toeplitz, Li96}.
See \cite{KRS97} for $W(A)$ when $N=3$.

\medskip

The present work is motivated by the applications of numerical range
and related material in quantum mechanics, especially
in the theory of quantum information \cite{KPLRS09,SHDHG08,GPMSCZ09}.
We will freely use the relevant physics terminology in what follows.
The set $\Omega_N$ of {\em pure quantum states}\footnote{These are the hermitian projections
of $\HH_N$ onto one-dimensional subspaces $\C|\psi\rangle$.}   
is naturally isomorphic to the complex projective space ${\C}P^{N-1}$. 
{\em Quantum states}\footnote{Also called {\em density matrices}.} are the operators 
$\rho\ge 0$ on ${\mathcal H}_N$ satisfying Tr$\rho=1$. The
set ${\mathcal Q}_N$ of quantum states is convex, and $\Omega_N\subset {\mathcal Q}_N$ is
the set of extremal points of ${\mathcal Q}_N$. The elements of ${\mathcal Q}_N\setminus\Omega_N$ are
{\em mixed quantum states}.

\medskip

For applications to quantum information it is crucial that
$W(A)$ is a plane projection of $\Omega_N$ \cite{DGHMPZ11}. If $N=2$ (i.e., the {\em one qubit case}), 
$\Omega_2={\C}P^1$  is the {\em Bloch sphere}
and ${\mathcal Q}_2\subset{\mathbf R}^3$ is the {\em Bloch ball}. The fact that a projection of the two-sphere 
is a (possibly degenerate) ellipse underlies the well known claims about
numerical ranges of $2\times 2$ matrices \cite{Li96}. The  {\em numerical shadow} of
an operator $A$ on ${\mathcal H}_N$ is a probability distribution $P_A(z)$ supported on $W(A)$ \cite{Z2009,GS10,DGHPZ11}.
Let $|\psi\rangle$ be distributed on the unit sphere $S({\mathcal H}_N)$ according to
the {\em Haar measure}. Then $P_A(z)$ is the probability density of $z=\langle \psi|A|\psi\rangle$. In the one
qubit case, $P_A(z)$ is the density of the plane shadow of the Bloch sphere under a light
beam~~\cite{Z2009}.

\medskip

Let $A_1, \dots A_m$ be hermitian operators on $\HH_N$. Their {\em joint numerical range}
(JNR) \cite{GJK04} is the set in $\R^m$ defined by
\begin{equation}  \label{JNR1}
W(A_1,\dots, A_m)\; = \;
\{(
\langle \psi|A_1|\psi\rangle, \dots,
\langle\psi|A_m|\psi\rangle\,):\langle\psi|\psi\rangle=1\}.
\end{equation}
Since $W(A_1,A_2)=W(A_1+iA_2)$, equation~~\eqref{JNR1} generalizes the notion of the numerical range
of a complex operator. We note that $W(A_1,\dots,A_m)$ is not necessarily
convex for $m>2$ \cite{GJK04}. For instance, the 
JNR of {\em Pauli matrices} is the Bloch sphere; see section~~\ref{joint}.

\medskip

We study the above notions and the relationships between them and quantum maps. 
Theorem~~\ref{jnr_thm} in section~~\ref{joint} shows that the JNR of an $m$-tuple
of hermitian operators is a linear projection of the set of pure quantum states to $\R^m$.
In section~~\ref{num_shad} we associate with any $m$-tuple
of hermitian operators a probability measure on ${\mathbf R}^m$. This is the {\em joint numerical shadow}
of the $m$-tuple of hermitian operators; it extends the concept of numerical shadow of a 
complex operator~~\cite{Z2009,GS10,DGHPZ11,P+12}. We point out a few basic properties
of joint numerical shadows, deferring a deeper study to a separate publication.

Let $\Phi$ be the {\em quantum map} on the set of quantum states defined by a $k$-tuple of 
{\em Kraus operators} satisfying  the {\em identity resolution}~~\eqref{ident_res_eq}.
In section~~\ref{maps} we study the effects of $\Phi$ on $m$-tuples of hermitian operators.
As Corollaries~~\ref{quant_map_cor} and~~\ref{num_range_cor} show, $\Phi$ shrinks the joint numerical ranges.

\medskip

Throughout the paper we emphasize the applications of our material in the theory of quantum information.
Examples 1, 2, 3, 4, 5 and 6 illustrate these  applications. For instance, examples 3 and 4 show
the shrinking of numerical ranges under particularly well known quantum maps in the qubit and the qutrit cases.

\section{Joint numerical ranges}    \label{joint}
Let $\LL_N$ (resp. $\SSS_N$,
$\PP_N$, $\Q_N$) denote the space of all (resp. hermitian, positive
definite, positive definite with trace $1$) linear operators on
${\HH}_N$. Let $\Omega_N\subset\Q_N$ be the set of rank one projections.
As vector spaces, $\LL_N=\C^{N^2}, \SSS_N=\R^{N^2}$. The scalar product
\begin{equation} \label{hilb_schm_eq}
(A,B)=\tr(AB^*),
\end{equation}
makes $\LL_N$ (resp. $\SSS_N\subset\LL_N$) a Hilbert space (resp. Euclidean space).
The set $\PP_N\subset\SSS_N$ is a
closed convex cone. Its interior consists of strictly positive operators,
$\rho>0$, and its apex is the zero operator. The set $\Q_N$ is
the intersection of  $\PP_N$ and the hyperplane $\{\Tr\rho=1\}$. It
is a bounded convex region (i.e., has nonempty
interior) in the $(N^2-1)$-dimensional affine space, and
$\Om_N\subset\Q_N$ is  the set of its extremal points.
For a unit vector $\psi\in\HH_N$ we set $\rho_{\psi}=|\psi\rangle\langle\psi|\in\Om_N$.
The manifold $\C P^{N-1}$ is the quotient of the unit sphere $S(\HH_N)$ by the natural
linear action of the unit circle $S^1\subset\C$.
The map $\psi\mapsto\rho_{\psi}$ yields an isomorphism of $\C P^{N-1}$ and $\Om_N$.
In what follows we will identify $\C P^{N-1}$ and $\Om_N$ via this isomorphism.

\medskip

Let $A_1,\dots,A_m\in\SSS_N$ be arbitrary. The mapping from $\C P^{N-1}$ to $\R^m$ given by
\begin{equation}    \label{jnr_map_eq}
{\rm jnr}_{(A_1,\dots,A_m)}:|\psi\rangle \mapsto (\langle\psi|A_1|\psi\rangle,\dots,\langle\psi|A_m|\psi\rangle).
\end{equation}
is the {\em joint numerical range map}. The range
$W(A_1,\dots,A_m)$ of this map is the {\em
joint numerical range} (JNR) of operators $A_1,\dots,A_m$. By the isomorphism
$\C P^{N-1}\simeq\Om_N$, we have  ${\rm jnr}_{(A_1,\dots,A_m)}:\Om_N\to\R^m$.

\medskip

We will recall a few basic notions in affine geometry. By a vector space we
will mean a finite dimensional real vector space. Let $V$ be a vector space.
A set $H\subset V$ is an {\em affine subspace} if there is a linear subspace $H_0\subset V$
and a vector $h_0\in H$ such that $H=H_0+h_0$. Let $G\subset U,H\subset V$ be affine subspaces.
Let $G_0\subset U,H_0\subset V$ be the corresponding linear subspaces. A map $Af:G\to H$ is an
{\em affine isomorphism} if there is a linear isomorphism  $Af_0:G_0\to H_0$ and vectors
$h_0\in H,g_0\in G$ so that
$$
Af(g)=Af_0(g-g_0)+h_0.
$$
\begin{defi}     \label{affine_equi_def}
Let $U,V$ be vector spaces, and let $X\subset U,Y\subset V$ be arbitrary sets.
They  are {\em affinely isomorphic} if there exist affine subspaces $H\subset
U,G\subset V$ such that $X\subset G,Y\subset H$, and an affine isomorphism
$A:G\to H$ such that $A(X)=Y$. The induced map $A:X\to Y$ is an affine isomorphism of $X$ onto $Y$.
\end{defi}
Note that linear isomorphism of sets are the special cases in the above setting
when the subspaces and the maps in question are, actually, linear.
We will not distinguish between the linear and affine situations in what follows.

\medskip

We will now expose a general topic in linear algebra. Let $U,V$ be vector spaces. We assume that
$U$ is a Euclidean space with the scalar product $(\cdot,\cdot)$. Let $m\ge 1$. With any nonzero
vectors $u_1,\dots,u_m\in U$ and $v_1,\dots,v_m\in V$ we associate a linear operator $L:U\to V$ by
\begin{equation}   \label{operator_eq}
L(u)=\sum_{i=1}^m(u,u_i)v_i.
\end{equation}
Let $\A\subset U$ (resp. $\B\subset V$) be the subspace spanned by $u_1,\dots,u_m$ (resp.
$v_1,\dots,v_m$). We will need a simple lemma about the operator
in equation~~\eqref{operator_eq}.\footnote{We leave the straightforward proof to the reader.}
\begin{lemma}  \label{main_lem}
1. Let $\dim\A=p\le m$. Assume without loss of generality that the vectors
$u_1,\dots,u_p$ span $\A$.
Then there are vectors $\tv_1,\dots,\tv_p\in\B$ such that the linear operator in equation~~\eqref{operator_eq}
satisfies
\begin{equation}   \label{operator_eq1}
L(u)=\sum_{i=1}^p(u,u_i)\tv_i.
\end{equation}
2. Let $\dim\B=q\le m$. Assume without loss of generality that the vectors
$v_1,\dots,v_q$ span $\B$. Then there are vectors $\tu_1,\dots,\tu_q\in\A$ such that
the linear operator in equation~~\eqref{operator_eq} satisfies
\begin{equation}   \label{operator_eq2}
L(u)=\sum_{i=1}^q(u,\tu_i)v_i.
\end{equation}
3.  If $p=m$ (resp. $q=m$) then $\tv_i=v_i$ (resp. $\tu_i=u_i$) for $1\le i \le m$.
\end{lemma}

The proposition below is immediate from Lemma~~\ref{main_lem}.
\begin{proposition}  \label{operator_prop}
Let the setting be as in Lemma~~\ref{main_lem}, and let $Pr_{\A}:U\to\A$ be the orthogonal projection.
Then i) There is a unique linear operator $M:\A\to\B$ such that $L=M\circ Pr_{\A}$;
ii) If $\dim\A =m$ (resp. $\dim\B =m$) then $M$ is surjective (resp. injective); iii)  If $\dim\A=\dim\B=m$
then $M$ is an isomorphism.
\end{proposition}

Let $G,H\subset U$ be affine subspaces in a vector space. We say that they are {\em parallel}
if $H=G+u_0$ for some $u_0\in U$.

\begin{corollary}  \label{operator_cor}
Let $U,V$ be vector spaces, and let $L:U\to V$ be as in equation~~\eqref{operator_eq}.
Let $\A\subset U,\B\subset V$ be the subspaces spanned by the vectors $u_1,\dots,u_m$
and $v_1,\dots,v_m$ respectively. Let $M:\A\to\B$ be the operator from Proposition~~\ref{operator_prop}.

Let $G\subset U$ be an affine subspace containing a subspace parallel to $\A$. Let $\Ga\subset G$ be an arbitrary set.
If $M$ is injective, then $Pr_{\A}(\Ga)$ and $L(\Ga)$ are affinely isomorphic.
\end{corollary}
\begin{proof}
Let $u_0\in U$ be such that $\A=G+u_0$. Let $\B_0\subset\B$ be the range of $M$.
Thus, $M:\A\to\B_0$ is a linear isomorphism.
By Proposition~~\ref{operator_prop}, we have
$$
L(\Ga+u_0)=M\,Pr_{\A}(\Ga+u_0)=M(\Ga+u_0)\subset\B_0.
$$
The mapping $u\mapsto L(u+u_0)$ induces an affine isomorphism of $G$ and $\B_0$.
\end{proof}

We will now apply the above material to joint numerical ranges.

\begin{thm}        \label{jnr_thm}
Let $A_1,\dots,A_m\in\SSS_N$ be traceless, linearly independent hermitian operators.
Let $\A\subset\SSS_N$ be the subspace spanned by them. Then i) The
joint numerical range of $A_1,\dots,A_m$ is affinely isomorphic to
$Pr_{\A}(\Om_N)$; ii) The convex hull of the joint numerical range
of $A_1,\dots,A_m$ is affinely isomorphic to $Pr_{\A}(\Q_N)$.
\end{thm}

\begin{proof}
Let $U=\SSS_N,V=\R^m$. Let $v_1,\dots,v_m$ be the standard basis in $\R^m$.
Then the map $jnr_{A_1,\dots,A_m}$ has the form equation~~\eqref{operator_eq} with $u_i=A_i$.
By Proposition~~\ref{operator_prop}, the map $M:\A\to\R^m$ is a linear isomorphism.
Set $G=\{X\in\SSS_N:\tr(X)=1\}$. Then $G\subset\SSS_N$ is an affine hyperplane containing the
affine subspace $\A+\1_N$. Claim i) now follows from  Corollary~~\ref{operator_cor}.

The convex hulls of affinely isomorphic sets are affinely isomorphic. Hence,
claim i) implies claim ii).
\end{proof}

\begin{corollary}        \label{jnr_cor}
Let the setting be as in Theorem~~\ref{jnr_thm}, with $m=N^2-1$.
Then $W(A_1,\dots,A_m)$ is affinely isomorphic to
$\Om_N$ and $\conv[W(A_1,\dots,A_m)]$ is affinely isomorphic to
$\Q_N$.
\end{corollary}
\begin{proof}
This is a special case of Theorem~~\ref{jnr_thm}. In this case
the affine hyperplane $G$ is parallel to $\A$. Hence $Pr_{\A}$ induces
an affine isomorphism of $\Om_N$ and $Pr_{\A}(\Om_N)$.
\end{proof}

\vspace{3mm}


\noindent{\bf Example 1}. Let $N=2$. The Pauli matrices
\begin{equation}   \label{pauli}
\sigma_1 =    \left[\begin{array}{cc}
                0  & 1 \\
                1  & 0
\end{array}\right], \ \
\sigma_2 =    \left[\begin{array}{cc}
                0  & -i \\
                i  & 0
\end{array}\right], \ \
\sigma_3 =    \left[\begin{array}{cc}
                1  &  0 \\
                0  & -1
\end{array}\right]
\end{equation}
form an orthogonal basis in the space of traceless Hermitian
$2\times 2$ matrices. By Corollary~~\ref{jnr_cor}, $W(\si_1,\si_2,\si_3)\subset\R^3$ is
the unit sphere. Indeed, this also holds by an elementary computation. Let $\psi=[z_1,z_2]$. 
Denote by $F:\C^2\to\R^3$ the joint numerical range map; (see equation~~\eqref{jnr_map_eq}).
Then 
$$
F(\psi)=[z_1\bar{z}_2+z_2\bar{z}_1,i(z_1\bar{z}_2-z_2\bar{z}_1),z_1\bar{z}_1+z_2\bar{z}_2],
$$
and the sum of squared coordinates is $(|z_1|^2+|z_2|^2)^2$.



Let $\rho$ be a $2\times 2$ density matrix. Its {\em
Bloch vector} $(\tau_1,\tau_2,\tau_3)=\tau\in{\mathbf R}^3$ is defined by the decomposition
\begin{equation}
\rho= \frac{1}{2} {\mathbf 1}_2 + \sum_{j=1}^3 \tau_j \sigma_j 
\label{bloch2}
\end{equation}
where $\tau_j  = \tr(\sigma_j\rho)/2$. Since $\rho=\rho^*$, we have $\tau\in{\mathbf R}^3$, 
and since $\rho \ge 0$
$$
\sqrt{\tau_1^2+\tau_2^2+\tau_3^2}=||\tau||\le 1/2.
$$
The equality $||\tau||=1/2$ holds if and only if $\rho$ is a pure
state. Thus, the map ${\rm jnr}_{(\sigma_1,\sigma_2,\sigma_3)}$
yields an isomorphism of $\Q_2$ (resp. $\Omega_2$) and the {\em
Bloch ball} (resp. {\em Bloch sphere}).

\medskip

Example 1 generalizes to $N>2$ as follows. Denote by $d=N^2-1$
the dimension of the space of traceles $N\times N$ hermitian matrices
with the scalar product~~\eqref{hilb_schm_eq}. Let
$\lambda_1,\dots,\lambda_d$ be an orthogonal, but not necessarily an orthonormal, basis. 
For instance, for $N=3$ the Gell--Mann matrices  $(\lambda_1,\dots,\lambda_8)$
satisfy the orthogonality relations $\tr(\lambda_j\lambda_k) = 2\delta_{jk}$ \cite{Schiff}.
Let $\rho$ be a $N\times N$ state. The counterpart of equation~~\eqref{bloch2} is
\begin{equation}
\rho= \frac{1}{N} {\mathbf 1}_N + \sum_{j=1}^{d} \tau_j \lambda_j.
\label{blochN}
\end{equation}
The generalized Bloch vector $\tau=\tau(\rho)$  has $d$
components $\tau_j  = {\rm Tr}(\rho \lambda_j)/{\rm Tr}(\lambda_j^2)$.
If the basis $\lambda_1,\dots,\lambda_d$ is orthonormal and $\rho_{\psi}=|\psi\rangle \langle \psi|$ is a pure state,
then $\tau_{\psi}=\tau(\rho_{\psi})$ satisfies $\tau_j  = \langle \psi |\lambda_j|\psi\rangle$.
Hence $\tau_{\psi}\in W(\lambda_1, \dots \lambda_d)$.

By Corollary~~\ref{jnr_cor}, the convex hull $\conv[W(\lambda_1, \dots \lambda_d)]$
is affinely isomorphic to the set $\Q_N$ of quantum states.
By Theorem \ref{jnr_thm}, if $\lambda_1, \dots, \lambda_m$
are linearly independent, then
the joint numerical range $ W(\lambda_1, \dots, \lambda_m)$
is affinely  isomorphic to a projection of $\Q_N$ to ${\mathbf R}^m$.
Compare with the results in~~\cite{GJK04,DGHMPZ11}.

\section{Quantum maps}      \label{maps}
We recall that the set $\LL_N$ of operators on ${\HH}_N$
is a Hilbert space with the scalar product given by equation~~\eqref{hilb_schm_eq}. 
Let $\LLL_N$ be the space of linear operators on $\LL_N$. For $\Phi\in\LLL_N$
we denote by $\Phi^*\in\LLL_N$ its adjoint.
For $X,Y\in\LL_N$ we define $\Phi_{X|Y}\in\LLL_N$ by
\begin{equation} \label{oper_hilb_schm_eq}
\Phi_{X|Y}(A)=XAY^*.
\end{equation}
Let $X_1,\dots,X_k\in\LL_N$ be arbitrary. We set
\begin{equation} \label{our_oper_eq}
\Phi_{(X_1,\dots,X_k)}=\sum_{i=1}^k\Phi_{X_i|X_i}.
\end{equation}
We will often use the notation $\Phi=\Phi_{(X_1,\dots,X_k)}$. In the physics literature the
transformations of $\PP_N$ defined by
\begin{equation}  \label{qmap}
\rho_1=\Phi(\rho)=\sum_{i=1}^k X_i \rho X_i^*.
\end{equation}
are called {\em quantum maps}. They
correspond to generalized quantum measurements with $k$ possible
outcomes \cite{BZ06}. Operators $X_i\in\LL_N$ are the
{\em measurement operators} or {\em Kraus operators}. Let $\Phi=\Phi_{(X_1,\dots,X_k)}$.
Then $\Phi_{(X_1^*,\dots,X_k^*)}=\Phi^*$, the adjoint operator with respect to the 
scalar product~~\eqref{hilb_schm_eq}. In the physics literature the quantum map
$\Phi^*$ is {\em dual to $\Phi$}. The duality
of quantum maps corresponds to the duality between the
Schr\"odinger and the Heisenberg representations in quantum
mechanics. In the former, the quantum states $\rho\in\PP_N$ evolve
via $\rho_1=\Phi(\rho)$, while the observables $A\in\SSS_N$ do
not. In the latter, the quantum states do not change and the
observables evolve by $A\mapsto A_1=\Phi^*(A)$.

\medskip

We will now establish a few basic properties of  quantum maps. The following lemma, whose proof is left
to the reader, will be used in Proposition~~\ref{misc_lem}.
\begin{lemma}   \label{basic_oper_prop}
{\em 1}. The operators $\Phi_{X|Y}$ span the vector space
$\LLL_N$.

\noindent{\em 2}. Let $X_1,Y_1,X_2,Y_2\in\LL_N$. Then
$$
\Phi_{X_2|Y_2}\Phi_{X_1|Y_1}=\Phi_{X_2X_1|Y_2Y_1}.
$$
\noindent{\em 3}. We have
$$
\left(\Phi_{X|Y}\right)^*=\Phi_{X^*|Y^*}.
$$
\end{lemma}
\begin{defi}     \label{ident_res_def}
We will denote by $\1\in\LL_N$ the identity operator.
Let $k\ge 1$. We say that the $k$-tuple of operators
$X_1,\dots,X_k\in\LL_N$ is an {\em identity resolution} if

\begin{equation}
\label{ident_res_eq}
\sum_{i=1}^k X_i X_i^*=\1.
\end{equation}
\end{defi}
\noindent The dual property
\begin{equation}
\label{dual_ident_res_eq}
\sum_{i=1}^k X_i^* X_i=\1,
\end{equation}
holds if and only if $X_1^*,\dots,X_k^*$ is an identity resolution.

\begin{propo}      \label{misc_lem}
Let $X_1,\dots,X_k\in\LL_N$ be arbitrary, and let
$\Phi=\Phi_{(X_1,\dots,X_k)}\in\LLL_N$. Then the following holds.

\noindent 1. The quantum map $\Phi:\LL_N\to\LL_N$ preserves
$\SSS_N$ and $\PP_N$.

\noindent 2. The operators $X_1,\dots,X_k$ satisfy
equation~~\eqref{ident_res_eq} if and only if $\Phi:\LL_N\to\LL_N$ preserves
$\1$.

\noindent 3. The operators $X_1,\dots,X_k$ satisfy
equation~~\eqref{dual_ident_res_eq} if and only if i) the map $\Phi$
preserves the trace; ii) the map $\Phi^*$ preserves $\1$.

\noindent 4. If $X_1,\dots,X_k$ satisfy
equation~~\eqref{dual_ident_res_eq}, then $\Phi$
preserves $\Q_N$.
\end{propo}
\begin{proof}
It suffices to prove 1) for operators $\Phi=\Phi_{X|X}$.
The former property is
immediate from $(XAX^*)^*=XA^*X^*$, and the latter from $\langle
XAX^*v,v\rangle=\langle AX^*v,X^*v\rangle$. Claim 2) is immediate
from the definition of $\Phi$. We have
$$
\Tr\left(\Phi(A)\right)=(\Phi(A),\1)=(A,\Phi^*(\1)).
$$
Claim 3) now follows from Lemma~~\ref{basic_oper_prop}.
Claim 4) is immediate from 1) and 3).
\end{proof}

\medskip

Set $\Phi=\Phi_{(X_1,\dots,X_k)}$. We say that the quantum map $\Phi$ is {\em unital} (resp. {\em trace
preserving}) if $\Phi(\1)=\1$ (resp. $\Tr\left(\rho(\Phi)\right)=\Tr\Phi$).
By Proposition~~\ref{misc_lem}, a quantum map is unital (resp. trace
preserving) if and only if equation~~\eqref{ident_res_eq} (resp.
equation~~\eqref{dual_ident_res_eq}) is satisfied. Since
$$
{\rm Tr} \left( \Phi(\rho) A \right)= {\rm Tr} \left( \sum_i X_i
\rho X_i^* A \right)={\rm Tr} \left(\sum_i \rho X_i^* A X_i
\right) = {\rm Tr}\left( \rho \Phi^*(A) \right),
$$
$\Phi$ is trace preserving
(resp. unital) if and only if $\Phi^*$ is unital (resp. trace preserving).


For $\psi\in\HH_N$ we set $\rho_{\psi}=|\psi\rangle \langle \psi|$.
Then $\rho_{\psi}\in\PP_N$, and for any $a\in\C$ we have
\begin{equation}   \label{scalar_eq}
\rho_{a\psi}=|a|^2\rho_{\psi}.
\end{equation}
In particular, $\rho_{\psi}=0$ if and only if  $\psi=0$ and $\rho_{\psi}\in\Q_N$ if and only if $||\psi||=1$.
\begin{proposition}     \label{quant_map_thm}
Let $X_1,\dots,X_k\in\LL_N$, and let
$\Phi=\Phi_{(X_1,\dots,X_k)}$.  Let $\psi\in\HH_N$ be any vector.
For $1\le i \le k$ such that $X_i\psi\ne 0$, set
$\psi_i=X_i\psi/||X_i\psi||$.\\

\noindent Then
\begin{equation}    \label{map_explicit_eq}
\Phi(\rho_{\psi})=\sum_i||X_i\psi||^2\rho_{\psi_i}.
\end{equation}
\end{proposition}
\begin{proof}
Let $X,Y\in\LL_N$ and $v\in\HH_N$. Then
$$
(\Phi_{X|Y}(\rho_{\psi}))v=X\langle \psi,Y^*v\rangle\psi=\langle
Y\psi,v\rangle X\psi.
$$
In particular, we have
\begin{equation}    \label{oper_eq}
\Phi_{X|X}(\rho_{\psi})=\rho_{X\psi}.
\end{equation}
The claim now follows from equations~~\eqref{our_oper_eq}
and~~\eqref{scalar_eq}.
\end{proof}

\medskip

\begin{corollary}       \label{quant_map_cor}
Let $X_1,\dots,X_k$ satisfy equation~~\eqref{dual_ident_res_eq}
and set  $\Phi=\Phi_{(X_1,\dots,X_k)}$. Let $\psi\in\HH_N$ be a
unit vector. Let $1\le m\le k$ be the number of unit vectors
$\psi_i$ from Proposition~~\ref{quant_map_thm}. Then
$$
\Phi(\rho_{\psi})=\sum_{i=1}^m p_i\rho_{\psi_i},
$$
where $p_1,\dots,p_m>0$ and $\sum p_i=1$.
\end{corollary}
\begin{proof}
Set $p_i=||X_i\psi||^2$. The claim follows from
equation~~\eqref{map_explicit_eq} and the identity
$$
\sum_i||X_i\psi||^2=\langle\,\left(\sum_{i=1}^k X_i^*
X_i\right)\psi,\psi\rangle.
$$
\end{proof}

%
\begin{corollary}   \label{num_range_cor}
Let $X_1,\dots,X_k\in\LL_N$, and set
$\Phi=\Phi_{(X_1,\dots,X_k)}$. Suppose that $X_1,\dots,X_k$
satisfy equation~~\eqref{ident_res_eq}. Then\\
\noindent 1. For any $A\in\LL_N$ we have
\begin{equation}      \label{incl_eq}
W(\Phi(A)) \ \subset \ W(A);
\end{equation}

\noindent 2. For any $A_1,\dots,A_m\in\SSS_N$ we have
\begin{equation}   \label{JNR_incl_eq1}
W(\Phi(A_1),\dots,\Phi(A_m)) \subset\conv(W(A_1,\dots,A_m));
\end{equation}
\noindent 3. If $W(A_1,\dots,A_m)$ is convex, then we have
\begin{equation}   \label{JNR_incl_eq2}
W(\Phi(A_1),\dots,\Phi(A_m)) \subset W(A_1,\dots,A_m).
\end{equation}
\end{corollary}
\begin{proof}
Let $A\in\LL_N$ be arbitrary, and let $\psi\in\HH_N$ be a unit
vector. Then
$$
\Tr(\Phi(A)\rho_{\psi})=\Tr\left((\sum_{i=1}^kX_iAX_i^*)\rho_{\psi}\right)=\Tr\left(A(\sum_{i=1}^kX_i^*\rho_{\psi}X_i)\right).
$$
Note that the operators $X_1^*,\dots,X_k^*$ satisfy
equation~~\eqref{dual_ident_res_eq}. By
Proposition~~\ref{quant_map_thm} and
Corollary~~\ref{num_range_cor}, there are unit vectors $\psi_j\in
\HH_N$ and probabilities $p_j$ such that
\begin{equation}   \label{incl_iden_eq}
\Tr\left(\Phi(A)\rho_{\psi}\right)=\sum_jp_j\Tr\left(A\rho_{\psi_j}\right)=\sum_jp_j\langle\psi_j|A|\psi_j\rangle.
\end{equation}
Claim 2 now follows from equation~~\eqref{jnr_map_eq}. Claim 3 is
a special case of Claim 2. Claim 3 and the Toeplitz-Hausdorff
theorem yield Claim 1.
\end{proof}

\medskip

The examples below illustrate the relationship between
quantum maps and numerical ranges, which is the subject of  
Corollary~~\ref{num_range_cor}.


\noindent{\bf Example 2.} Let $N=2$. The {\em decaying channel}
is the discrete dynamics $\Phi=\Phi_{X_1,X_2}$ corresponding to the
Kraus operators
$X_1=\left[\begin{array}{cc}
                  1  &  0  \\
                  0  & \sqrt{1-p}   \\
\end{array}\right]$ and
$X_2= \left[\begin{array}{cc}
                  0  & \sqrt{p}  \\
                  0  & 0   \\
\end{array}\right]$, where $p \in (0,1)$ is a free parameter. Set $A^{(j)}=\Phi^{j-1}(A^{(1)})$, where
$A^{(1)}=
 \left[\begin{array}{cc}
                  0  & 1  \\
                  0  & 0   \\
\end{array}\right]$. Note that $\tr(A^{(1)})=0$. The set $W(A^{(1)})$ is the disc of radius $1/2$
centered at the origin, i.e., at the {\em barycenter point} $\tr(A^{(1)})/2$.
We have $A^{(j)}=(1-p)^{(j-1)/2}A^{(1)}$, and $W(A^{(j)})=(1-p)^{(j-1)/2}W(A^{(1)})$ is the disc
of radius $\frac12 (1-p)^{(j-1)/2}$ with the center at $0$.
For instance, $A^{(2)}=\Phi(A^{(1)})=
 \left[\begin{array}{cc}
                  0  & \sqrt{1-p}  \\
                  0  & 0   \\
\end{array}\right]$, and
$A^{(3)}=\Phi^2(A^{(1)})=
 \left[\begin{array}{cc}
                  0  &  1-p  \\
                  0  & 0   \\
\end{array}\right]$.
The limit of $W(A^{(j)})$, as $j\to \infty$, is $\{0\}$.
Fig. 1a shows $W(A^{(1)}),W(A^{(2)}),W(A^{(3)})$ for $p=0.5$.

\begin{figure}[htbp]
\begin{center}
\includegraphics[scale=0.55]{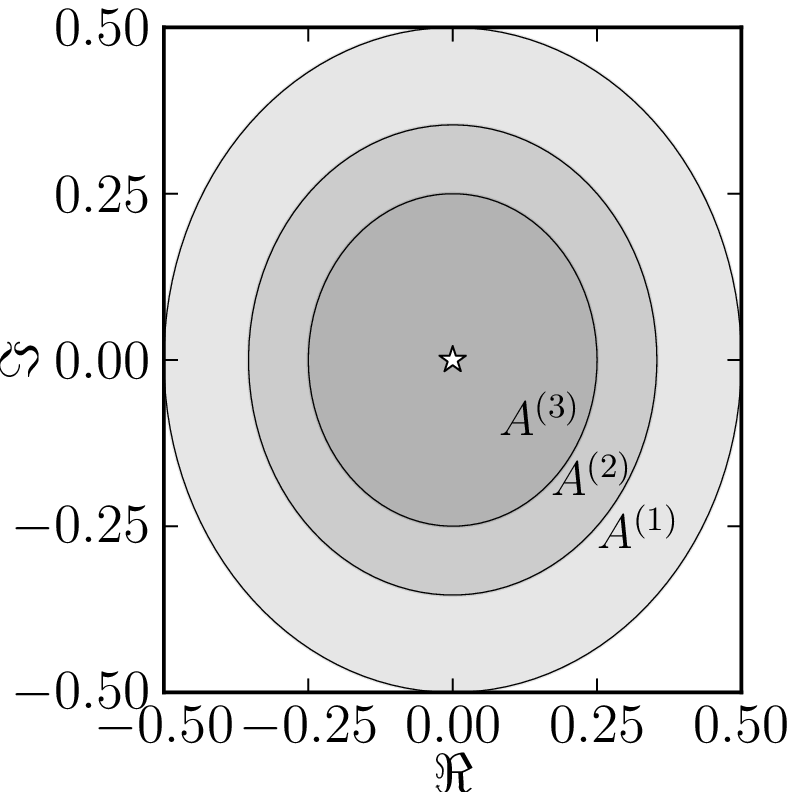}
\hskip 0.6cm
\includegraphics[scale=0.55]{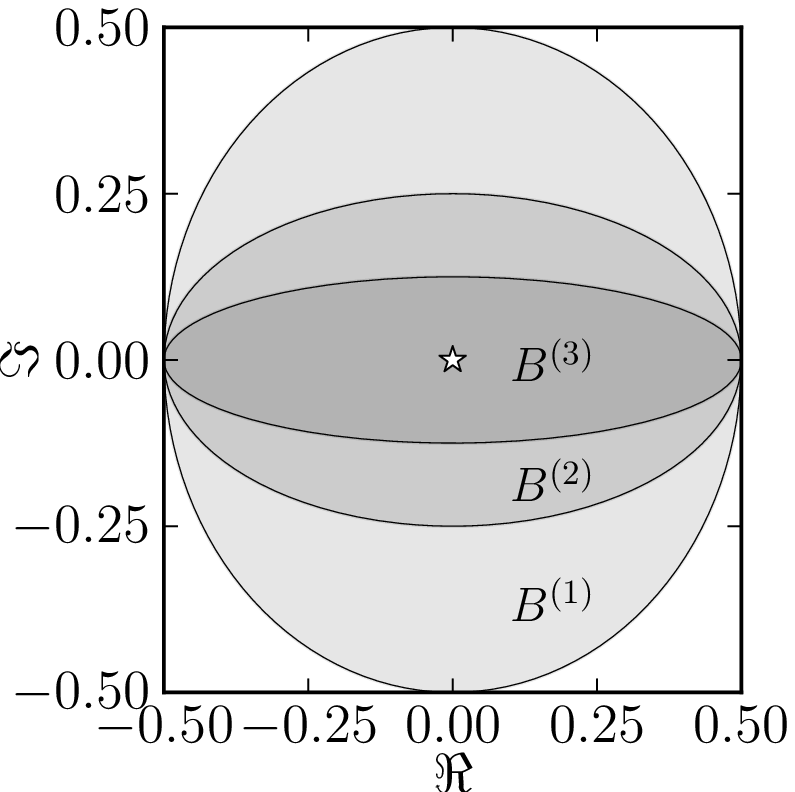}
\caption{Numerical ranges $W(A^{(1)}),W(A^{(2)}),W(A^{(3)})$ from
Example 2 for $p=0.5$ and $W(B^{(1)}),W(B^{(2)}),W(B^{(3)})$ from
Example 3 for $p=0.25$. The stars mark the barycenter points in these
examples. They are preserved by the dynamics.
}
\label{fig1}
\end{center}
\end{figure}


\noindent{\bf Example 3.} Let again $N=2$. The {\em phase-flip channel}
is the discrete dynamics $\Psi=\Phi_{X_1,X_2}$ corresponding to the
Kraus operators $X_1=\sqrt{1-p}\; {\1}_2$ and $X_2=\sqrt{p}\sigma_1$.
Here again, $p \in (0,1)$ is a free parameter.
We set $B^{(1)}=A^{(1)}$ and $B^{(j+1)}=\Psi^j(B^{(1)})$. Then
$B^{(2)}=
 \left[\begin{array}{cc}
                  0  & 1-p  \\
                  p  & 0   \\
\end{array}\right]$
and
$B^{(3)}=
 \left[\begin{array}{cc}
                  0  &  1-2p(1-p)  \\
                  2p(1-p)  & 0   \\
\end{array}\right]$. The numerical ranges of all $B^{(j)}$ are ellipses.
Fig. 1b  shows $W(B^{(1)}),W(B^{(2)}),W(B^{(3)})$ for $p=0.25$.

\medskip

\noindent{\bf Example 4.} This example is a generalization of Example 3 to
$N=3$. The {\em double flip channel acting on a qutrit} is the the discrete dynamics
$\Xi=\Phi_{X_1,X_2,X_3}$; the Kraus operators are $X_1=\sqrt{1-p-q}\1_3$,
$X_2=
\sqrt{p} \left[\begin{array}{ccc}
                  0  &  1 & 0 \\
                  1  &  0 &   0   \\
                  0  & 0  & 1  \\
\end{array}\right]$
and
$X_3= \sqrt{q}
\left[\begin{array}{ccc}
                  1  &  0  & 0 \\
                  0  &  0  & 1 \\
                  0  &  1  & 0 \\
\end{array}\right]$. The parameters $p,q$ satisfy $0\le p,q,\,p+q\le 1$.
The operators $X_2$ and $X_3$ correspond to bit flips with probabilities $p$ and $q$.
The operator $\Xi$ is trace preserving. The operator
$C^{(1)}=
 \left[\begin{array}{ccc}
                  0  &  1  & 0 \\
                  0  &  1  & 0 \\
                  0  &  0  & 2i \\
\end{array}\right]$ was studied in \cite{P+12}; the
numerical range $W(C^{(1)})$ is an ellipse. The {\em barycenter} of $W(C^{(1)})$ is $\tr(C^{(1)})/3=(1+2i)/3$.
Set $C^{(j)}=\Xi^{j-1}(C^{(1)})$. For instance,
$C^{(2)}=
 \left[\begin{array}{ccc}
                  p  &  1-p-q  &       q \\
                  p  &  1-p-q(1-2i)  & 0 \\
                  0  &  0  & 2i +q(1-2i) \\
\end{array}\right]$. In Figure~~\ref{fig2} $p=0.5$ and $q=0.4$, hence $1-p-q=0.1$.

\begin{figure}[htbp]
\begin{center}
\includegraphics[scale=0.65]{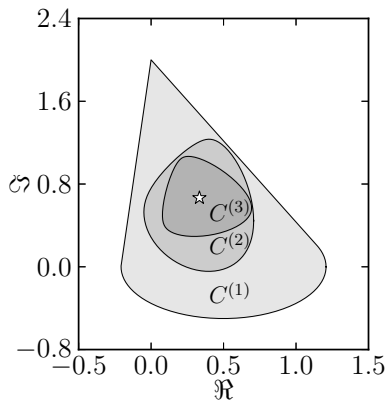}
\caption{This figure illustrates Example 4 and the inclusion relation~~\eqref{incl_eq}.
It shows the numerical ranges $W(C^{(1)}),W(C^{(2)}),W(C^{(3)})$ for $p=0.5$ and $q=0.4$.
The limit of $W(C^{(j)})$, as $j\to\infty$, is the barycenter point $(1+2i)/3$
 marked by the star.
}
 \label{fig2}
\end{center}
\end{figure}

\section{Joint numerical shadows}       \label{num_shad}
We will first recall the notion of the {\em numerical shadow} of an operator on $\HH_N$ \cite{DGHPZ11,GS10}.
Let $\mu$ be the normalized Haar measure on $S(\HH_N)=\{\psi\in\HH_N:||\psi||=1\}$, i.e.,
the unit sphere $S^{2N-1}$. We also denote by $\mu$ its push-forward to $\Om_N\simeq\C
P^{N-1}$.\footnote{It is known as the Fubini--Study measure in the
physics literature.} Let $A\in\LL_N$. The numerical shadow $\nu_A$ is the push-forward of $\mu$ to
$\C$ under the numerical range map. If $dz$ denotes the Lebesgue
measure on $\C$, then $d\nu_A(z)=P_A(z)dz$ where
\begin{equation}  \label{shadow2}
P_A(z) = \ \int_{\Omega_N} {\rm d} \mu(\psi) \
 \delta\Bigl( z-\langle \psi|A|\psi\rangle\Bigr)
\end{equation}
is a probability distribution.

\medskip

Let now $A_1,\dots, A_m\in\SSS_N$. Their {\em joint numerical shadow} $\nu_{A_1,\dots,A_m}$ is
the push-forward of $\mu$ under the joint numerical range map ${\rm jnr}_{(A_1,\dots,A_m)}:\Om_N\to\R^m$. 
Thus, $\nu_{A_1,\dots,A_m}$ is a probability measure supported on $W(A_1,\dots,A_m)$. Let $dx^m$ denote the Lebesgue
measure on $\R^m$. Then 
$d\nu_{A_1,\dots,A_m}=P_{A_1,\dots,A_m}(x_1,\dots,x_m)dx^m$ where
\begin{equation}    \label{shadowk}
P_{A_1,\dots,A_m}(x_1,\dots,x_m) = \int_{\Omega_N} {\rm d} \mu(\psi) \
 \prod_{j=1}^m  \delta\Bigl(x_j-\langle \psi|A_j|\psi\rangle\Bigr).
\end{equation}
is a probability distribution.
Numerical shadow is the special case of the joint numerical shadow
corresponding to $m=2$. See the examples below for illustration.

\medskip

Let $k_1,\dots,k_m\in\N$. The {\em moments}
\begin{equation}    \label{moments}
I_{k_1,\dots,k_m}(A_1,\dots,A_m)=\int_{\R^m}x_1^{k_1}\cdots x_m^{k_m}d\nu_{A_1,\dots,A_m}(x_1,\dots,x_m).
\end{equation}
are well defined and, by the Stone-Weierstrass theorem, uniquely determine the joint numerical shadow.
When $m=2$, we recover the moments of the numerical shadow $\nu_A$ for the matrix $A=A_1+iA_2\in\LL_N$
introduced in~~\cite{DGHPZ11}.  Some of the results  in~~\cite{DGHPZ11} extend to the 
moments of joint numerical shadows
for arbitrary $m$. We will report on this in a separate publication.

\medskip

If $\nu$ is a measure on $\R^m$ and $a\in\R$, we denote by $\nu^a$
the push-forward of $\nu$ under the self-map $v\mapsto av$ of
$\R^m$. If $\nu_1,\nu_2$ are measures on $\R^m$, we denote by
$\nu_1*\nu_2$ their {\em convolution}. The following proposition
exposes a few basic properties of joint numerical shadows. We
leave the proof to the reader.

\begin{proposition}      \label{jnr_lem}
1. Let $A_1,\dots,A_m\in\SSS_N$. Let $U\in\LL_N$ be a unitary
operator. Set $B_i=UA_iU^*,1\le i \le m$. Then
$$
\nu_{B_1,\dots,B_m}=\nu_{A_1,\dots,A_m}.
$$

\noindent 2. Let $A_1,\dots,A_m\in\SSS_N$ and
$B_1,\dots,B_m\in\SSS_N$ be arbitrary. Let $a,b\in\R$. Then
$$
\nu_{aA_1+bB_1,\dots,aA_m+bB_m}=\nu_{A_1,\dots,A_m}^a*\nu_{B_1,\dots,B_m}^b.
$$
\end{proposition}

\medskip

\noindent{\bf Example 5}. We review Example 1. The
set $W(\sigma_1,\sigma_2,\sigma_3)\subset\R^3$ is the unit sphere.
The joint numerical shadow
$\nu_{\sigma_1,\sigma_2,\sigma_3}$ is the normalized Haar measure.

\medskip

\noindent{\bf Example 6}. We use the isomorphism
$\HH_4=\HH_2\otimes\HH_2$ to define the extensions of Pauli
matrices $A_j=\sigma_j\otimes {\mathbf 1}_2,\,j=1,2,3$. We use the
fact that the Haar measure on the space of pure states in ${\HH}_A
\otimes{\HH}_B$ induces by partial trace, $\omega= {\rm Tr}_B
|\psi\rangle \langle \psi|$, the Lebesgue measure on the space of
mixed states in $\HH_2$ \cite{ZS00}. Using the equality between
the expected values of an operator on $\HH_2$ and the extended
operator on $\HH_4$, we obtain that $\nu_{A_1,A_2,A_3}$ is the
normalized Lebesgue measure on the  Bloch ball.
Let $B_j={\mathbf 1}_2\otimes\sigma_j,\,j=1,2,3$. The
{\em swap operator}
\begin{equation}
S = \left[\begin{array}{cccc}
                  1  & 0 & 0 & 0 \\
                  0  & 0 & 1 & 0 \\
                  0  & 1 & 0 & 0 \\
                  0  & 0 & 0 & 1 \\
\end{array}\right],
\label{swap}
\end{equation}
is a unitary operator on $\HH_4$ satisfying $B_j=SA_jS^*$. By
Proposition~~\ref{jnr_lem} and the above discussion,
$\nu_{B_1,B_2,B_3}$ is the normalized Lebesgue measure on the
Bloch ball.

\bigskip

\noindent{\bf Acknowledgements}. This project started at the 43th Symposium on Mathematical Physics in
Toru{\'n}; it was accomplished during the 44th Symposium.
It is a pleasure to thank the organizers of these Symposia. We are obliged to
Piotr Gawron for creating the figures.
K.~{\.Z}. is grateful to C.~F.~Dunkl, J.~A.~Holbrook, J.~Miszczak
and Z.~Pucha{\l}a for fruitful discussions on numerical shadows
and acknowledges financial support by the grant N202 090239
of the Polish Ministry of Science and Higher Education. The work of
E.G. was partially supported by the MNiSzW grant N N201 384834 and
the NCN Grant DEC-2011/03/B/ST1/00407. He
acknowledges stimulating discussions with Bent Orsted in Aarhus in
June 2011. Finally, we thank the anonymous referee for comments and suggestions.


\begin{thebibliography}{99}

\bibitem{HJ2}
A.~Horn and C.~R. Johnson,
\newblock {\sl Topics in Matrix Analysis},
\newblock Cambridge University Press, Cambridge, 1994.

\bibitem{GR97}
K.~E. Gustafson and D.~K.~M. Rao.
\newblock {\sl Numerical Range: The Field of Values of Linear Operators and
  Matrices}.
\newblock Springer-Verlag, New York, 1997.

\bibitem{Toeplitz} O. Toeplitz, {\em Das algebraische Analogon zu einem
Satze von Fej\'er}, Math. Zeitschrift {\bf 2} (1918), 187--197.

\bibitem{Hausdorff} F. Hausdorff, {\em Der Wertvorrat einer
Bilinearform}, Math. Zeitschrift {\bf 3} (1919), 314--316.

\bibitem{Gu04} E. Gutkin,
The Toeplitz-Hausdorff theorem revisited: relating linear algebra and geometry,
{\sl Math. Intelligencer} {\bf 26}, 8 -- 14 (2004).

\bibitem{JAG98} E. Jonckheere, F. Ahmad, E. Gutkin,
Differential topology of numerical range,
{\sl Lin. Alg. Appl.} {\bf  279}, 227 -- 254 (1998).

\bibitem{Li96} C.~K. Li,
A simple proof of the elliptical range theorem,
{\sl Proc. Am. Math. Soc.} {\bf 124}, 1985-1986 (1996).

\bibitem{KRS97} D.S. Keeler, L. Rodman and I.M. Spitkovsky,
The numerical range of $3 \times 3$ matrices,
{\sl Lin. Alg. Appl.} {\bf 252}, 115-1139 (1997).

\bibitem{KPLRS09} D.~W. Kribs, A.~Pasieka, M.~Laforest, C.~Ryan, and M.~P. Silva,
\newblock Research problems on numerical ranges in quantum computing,
\newblock {\sl Linear and Multilinear Algebra} {\bf 57}, 491-502 (2009).

\bibitem{SHDHG08}
T.~Schulte-Herbr\"uggen, G.~Dirr, U.~Helmke, and S.~J. Glaser,
\newblock The significance of the $c$-numerical range and the local
  {\emph{c}}-numerical range in quantum control and quantum information,
{\sl Linear and Multilinear Algebra} {\bf 56}, 3-26  (2008)

\bibitem{GPMSCZ09} P. Gawron, Z. Pucha{\l}a, J.~A.~Miszczak,
 {\L}.~Skowronek,  M.-D.~Choi, and K.~{\.Z}yczkowski,
Restricted numerical range: a versatile tool in the theory of quantum information,
{\sl J. Math. Phys.} {\bf 51}, 102204 (24pp) (2010).

\bibitem{DGHMPZ11} C.F. Dunkl, P. Gawron, J.A. Holbrook, J. Miszczak,
Z.~Pucha{\l}a and  K.~{\.Z}yczkowski,
Numerical shadow and geometry of quantum states,
{\sl J. Phys.}   {\bf A 44}, 335301 (19pp) (2011).

\bibitem{GS10} T. Gallay and D. Serre,
Numerical measure of a complex matrix,
{\sl Commun. Pure Apl. Math.} {\bf 65}, 287-336 (2012).

\bibitem{DGHPZ11} C.F. Dunkl, P. Gawron, J.A. Holbrook, Z. Pucha{\l}a and  K. {\.Z}yczkowski,
 Numerical shadows: Measures and densities on the numerical range,
{\sl Lin. Algebra Appl} {\bf  434}, 2042-2080 (2011).

\bibitem{Z2009} K. \.Zyczkowski (with M.--D. Choi, C. Dunkl, J. Holbrook, P. Gawron,
J.Miszczak, Z. Pucha{\l}a, and L. Skowronek),
Generalized numerical range as a versatile tool to study quantum entanglement,
{\sl Oberwolfach Report} {\bf 59}, 34--37 (2009).

\bibitem{P+12}  Z.~Pucha{\l}a, J. A. Miszczak, P. Gawron,
 C. F. Dunkl, J. A. Holbrook,  and  K.~{\.Z}yczkowski,
Restricted numerical shadow and geometry of quantum entanglement,
{\sl J. Phys.} {\bf A 45}, 415309 (28pp) (2012).

\bibitem{GJK04} E. Gutkin, E.A. Jonckheere, M. Karrow,
Convexity of the joint numerical range:
topological and differencial geometric viewpoints,
{\sl Lin. Alg. Appl.} {\bf  376}, 143 -- 171 (2004).

\bibitem{BZ06} I. Bengtsson and K. \.Zyczkowski,
{\sl Geometry of Quantum States}, Cambridge UP, Cambridge, 2006.

\bibitem{Schiff} L. I. Schiff, {\sl Quantum Mechanics},
New York, McGraw-Hill 1968.

\bibitem{ZS00} K. {\.Z}yczkowski and H.-J. Sommers,
 Induced measures in the space of mixed quantum states,
{\sl J. Phys.} {\bf A 34}, 7111-7125 (2001).

\end{thebibliography}
\end{document}